\documentclass[a4paper,USenglish]{scrartcl}
\usepackage{a4wide}
\usepackage[utf8]{inputenc}
\usepackage{amsmath, amsthm, amssymb}
\usepackage{enumitem}
\usepackage{hyperref} 
\newtheorem{thm}{Theorem}[section]

\newtheorem{lem}[thm]{Lemma}

\newtheorem{prop}[thm]{Proposition}

\theoremstyle{remark}

\theoremstyle{definition}
\newtheorem{defi}[thm]{Definition}

\newcommand{\p}{\mathsf{P}}
\newcommand{\sP}{\mathsf{\#P}}
\newcommand{\PP}{\mathsf{PP}}
\newcommand{\FP}{\mathsf{FP}}
\newcommand{\coRP}{\mathsf{coRP}}
\newcommand{\RP}{\mathsf{RP}}
\newcommand{\BPP}{\mathsf{BPP}}
\newcommand{\NP}{\mathsf{NP}}

\newcommand{\coNP}{\mathsf{coNP}}
\newcommand{\CH}{\mathsf{CH}}

\newcommand{\pn}[1]{\textsc{#1}}

\newcommand{\problem}[3]{
\vspace{\topsep}
\noindent\fbox{\begin{minipage}{.8\textwidth}
\pn{#1}\\
 \textbf{Input:} #2\\
\textbf{Problem:} #3
\end{minipage}
}
\vspace{\topsep}
}

\newcommand{\cc}{\mathsf{C}}
\newcommand{\ce}{\mathsf{C_=}}
\newcommand{\cne}{\mathsf{C_{\ne}}}

\newcommand{\lal}{\ensuremath{\alpha}}
\newcommand{\lbe}{\ensuremath{\beta}}
\newcommand{\lga}{\ensuremath{\gamma}}

\newcommand{\lep}{\ensuremath{\epsilon}}

\newcommand{\ltau}{\ensuremath{\tau}}

\newcommand{\abs}[1]{\lvert#1\rvert}
\newcommand{\set}[1]{{ \left\{#1\right\} }}

\title{Monomials in arithmetic circuits: Complete problems in the counting hierarchy}

\begin{document}

\author{Hervé Fournier and Guillaume Malod\\ 
Univ Paris Diderot, Sorbonne Paris Cité,\\ 
Institut de Mathématiques de Jussieu, UMR 7586 CNRS, \\F-75205 Paris, France\\
{\small \texttt{\{fournier,malod\}@math.univ-paris-diderot.fr}}
\and
Stefan Mengel\footnote{Partially supported by DFG grants BU 1371/2-2 and BU 1371/3-1.}\\Institute of Mathematics\\ University of Paderborn\\ D-33098 Paderborn, Germany\\ {\small\texttt{smengel@mail.uni-paderborn.de}} 
}

\maketitle
\begin{abstract}
We consider the complexity of two questions on polynomials given by arithmetic circuits: testing whether a monomial is present and counting the number of monomials. We show that these problems are complete  for  subclasses of the counting hierarchy which had few or no known  natural complete problems before.
 We also study these questions for circuits computing multilinear polynomials.
\end{abstract}

\section{Introduction}

Several recent papers in  arithmetic circuit complexity  refer to a family of classes called the counting hierarchy consisting of the classes $\PP \cup \PP^\PP \cup \PP^{\PP^\PP}\cup \ldots$.  For example, Bürgisser~\cite{Burgisser09} uses these classes  to connect computing integers to computing polynomials, while   Jansen and Santhanam~\cite{JS11} --- building on results by Koiran and Perifel \cite{KP11} ---  use them to derive lower bounds from derandomization. This hierarchy was originally introduced by
Wagner~\cite{wagner86} to classify the complexity of combinatorial problems. Curiously, after Wagner's paper and another by Tor{\'a}n~\cite{toran88}, this original motivation of the counting hierarchy has to the best of our knowledge not been pursued for more than twenty years. Instead, research focused on structural properties and the connection to threshold circuits~\cite{AW93}.  As a result, there are very few natural  complete problems for classes in the counting hierarchy:  for instance, Kwisthout et al.~give in~\cite{kbvdg11}  ``the first problem with a practical application that is shown to be $\FP^{\PP^\PP}$\hspace*{-1mm}-complete''.
 The related class $\ce\p$ appears to have no natural complete problems at all (see~\cite[p.~293]{hemaspaandraO02}). 
It is however possible to define generic complete problems
 by starting with a  $\sP$-complete problem and considering the variant where an instance and a positive integer are provided and the question is to decide whether the number of solutions for this instance is equal to the integer. We consider these problems to be counting problems disguised as decision problems and thus not as natural complete problems for $\ce\p$,
in contrast to the questions studied here.
 Note that  the corresponding logspace counting class $\ce\mathsf{L}$  is known to have interesting complete problems from linear algebra~\cite{ABO99}. 
 
 In this paper we follow Wagner's original idea and show that the counting hierarchy is a helpful tool to classify the complexity of several natural problems on arithmetic circuits by showing complete problems for the classes $\PP^\PP$, $\PP^\NP$ and $\ce\p$.\footnote{Observe that Hemaspaandra and Ogihara~\cite[p.~293]{hemaspaandraO02} state that Mundhenk et al.~\cite{mundhenk00} provide natural complete problems for $\PP^\NP$. This appears to be a typo as Mundhenk et al.~in fact present complete problems not for $\PP^\NP$ but for the class $\NP^\PP$ which indeed appears to  have several interesting complete problems in the AI/planning literature.}
The common setting of these problems is the use of circuits or straight-line programs to represent polynomials. Such a representation can be much more efficient than giving the list of  monomials, but  common operations on polynomials may become more difficult. An important example is the question of determining whether the given polynomial is identically zero. This is easy to do when given a list of monomials. When the polynomial is given as a circuit, the problem, called ACIT for \emph{arithmetic circuit identity testing}, is solvable in $\coRP$ but  is not known to be in $\p$. In fact, derandomizing this problem would imply circuit lower bounds, as shown in~\cite{ki04}.
This question
thus plays a crucial part in complexity and it is natural to consider other problems on polynomials represented as circuits.
In this article we consider mainly two questions. 
 
The first question, called \pn{ZMC} for \emph{zero monomial coefficient}, is to decide whether a given monomial in a circuit has coefficient $0$ or not. This  problem has already been studied
 by Koiran and Perifel~\cite{koiranP07}. They showed that when the  formal degree of the circuit is polynomially bounded the problem is complete for $\p^{\sP}$. Unfortunately this result is not fully convincing, because it is formulated with the rather obscure  notion of strong nondeterministic Turing reductions.  We remedy this situation by proving a completeness result for the class $\ce\p$ under more traditional logarithmic space many-one reductions. 
This provides a  natural complete problem for this class. 
Koiran and Perifel also considered the general case  of \pn{ZMC}, where the formal degree of the circuits is not bounded. They showed that \pn{ZMC} is in 
$\CH$. We  provide  a better upper bound by proving that \pn{ZMC} is in $\coRP^\PP$.
 We finally study  the  case of monotone circuits and show that the problem is then $\coNP$-complete.
  
 The second problem is to count the number of monomials in the polynomial computed by a circuit. This seems like a natural question whose solution should not be too hard, but in the general case it turns out to be $\PP^\PP$-complete, 
and the hardness holds even for weak circuits. We thus obtain another natural complete problem, in this case for the second level of the counting hierarchy.

Finally, we study the two above problems in the case of circuits computing multilinear polynomials.
We show that our first problem becomes equivalent to the fundamental problem  \pn{ACIT} and
that counting monomials becomes $\PP$-complete. 
 
\section{Preliminaries}\label{sec:prem}

\subparagraph*{Complexity classes}

We assume the reader to be familiar with basic concepts of computational complexity theory (see e.g.~\cite{AB09}). All reductions in this paper will be logspace many-one unless stated otherwise.

We consider different counting decision classes in the counting hierarchy~\cite{wagner86}. These classes are defined analogously to the quantifier definition of the polynomial hierarchy but, in addition to the quantifiers $\exists$ and $\forall$, the quantifiers $\cc$, $\ce$ and $\cne$ are used.

\begin{defi}
 Let $\mathcal{C}$ be a complexity class.
\begin{itemize}
 \item $A\in \cc\mathcal{C}$ if and only if there is $B\in \mathcal{C}$, $f\in \mathsf{FP}$ and a polynomial $p$ such that 
\[x \in A \Leftrightarrow \left| \left\{ y\in \{0,1\}^{p(|x|)} \mid (x,y)\in B \right\} \right| \ge f(x),
\]
\item $A\in \ce\mathcal{C}$ if and only if there is $B\in \mathcal{C}$, $f\in \mathsf{FP}$ and a polynomial $p$ such that 
\[x \in A \Leftrightarrow \left| \left\{ y\in \{0,1\}^{p(|x|)} \mid (x,y)\in B \right\} \right| = f(x),
\]
\item $A\in \cne\mathcal{C}$ if and only if there is $B\in \mathcal{C}$, $f\in \mathsf{FP}$ and a polynomial $p$ such that 
\[x \in A \Leftrightarrow \left| \left\{ y\in \{0,1\}^{p(|x|)} \mid (x,y)\in B \right\} \right| \ne f(x).
\]
\end{itemize}
\end{defi}

Observe that $\cne\mathcal{C} = \mathsf{co}\ce\mathcal{C}$ with the usual definition $\mathsf{co}\mathcal{C} = \{ L^c\ |\ L \in \mathcal{C}\}$, where $L^c$ is the complement of $L$. That is why the quantifier $\cne$ is often also written as $\mathsf{coC_=}$, so $\cne \p$ is sometimes called $\mathsf{coC_=P}$. 

The counting hierarchy $\mathsf{CH}$ consists of the languages from all classes that we can get from $\mathsf{P}$ by applying the quantifiers $\exists$, $\forall$, $\cc$, $\ce$ and $\cne$ a constant number of times. Observe that with the definition above $\mathsf{PP} = \cc\mathsf{P}$. Tor{\'a}n~\cite{toran91} proved that this connection between $\mathsf{PP}$ and the counting hierarchy can be extended and that there is a characterization of $\mathsf{CH}$ by oracles similar to that of the polynomial hierarchy. We  state some such characterizations which we will need later on,
followed by other technical lemmas.

\begin{lem}\cite{toran91}\label{lem:quantorsvsoracles}
$\PP^\NP = \cc\exists\p$.
\end{lem}

\begin{lem}
\label{lem:oraclesPPPP}
$\PP^\PP = \cc\cne\p$.
\end{lem}

\begin{proof}[Proof of  Lemma~\ref{lem:oraclesPPPP}]
 This is not stated in~\cite{toran91}  nor is it a direct consequence, because Tor{\'a}n does not consider the $\cne$-operator. It can be shown with similar techniques and we give a proof for completeness. We show that $\cc\cne \p = \cc\cc \p$, the claim then follows, because $\cc\cc \p = \PP^\PP$ by~\cite{toran91}.

The direction from left to right is straightforward: From the definition we have $\cc\cne\p \subseteq \PP^\sP$. By binary search we have $\PP^\sP = \PP^\PP = \cc\cc\p$.
The other direction is a little more work: Let $L\in \cc\cc\p$. There are $A\in \p$,$f,g \in \FP$ and a polynomial $p$ such that 
\begin{eqnarray}
x\in L& \Leftrightarrow & \text{there are more than $f(x)$ values $y\in \{0,1\}^{p(|x|)}$ such that }\notag \\ && \quad \left| \left\{ z\in \{0,1\}^{p(|x|)} \mid (x,y,z)\in A \right\} \right| \ge g(x,y) \notag \\
&\Leftrightarrow& \text{there are more than $f(x)$ values $y\in \{0,1\}^{p(|x|)}$ such that } \notag \\ && \quad \forall v\in \{1,\ldots , 2^{p(|x|)}\}\colon \left| \left\{ z\in \{0,1\}^{p(|x|)} \mid (x,y,z)\in A \right\} \right| \ne g(x,y) - v \label{eq:1}\\
&\Leftrightarrow& \text{there are more than $2^{p(|x|)} (2^{p(|x|)} - 1)+f(x)$ pairs $(x,v)$ with}\notag \\&&\quad \text{ $y\in \{0,1\}^{p(|x|)}$ and $v\in \{1,\ldots , 2^{p(|x|)}]\}$ such that } \notag \\ && \quad \left| \left\{ z\in \{0,1\}^{p(|x|)} \mid (x,y,z)\in A \right\} \right| \ne g(x,y) - v \label{eq:2}.
\end{eqnarray}

From statement (\ref{eq:2}) we directly get $L\in \cc\cne \p$ and thus the claim. To see the last equivalence we define $r(x,y):=\left|\left\{ z\in \{0,1\}^{p(|x|)} \mid (x,y,z)\in A \right\} \right|$. Fix $x,y$, then obviously $r(x,y)\ne g(x,y)-v$ for all but at most one $v$. It follows that of the pairs $(y,v)$ in the last statement $2^{p(|x|)} (2^{p(|x|)} - 1)$ always lead to inequality. So statement (\ref{eq:2}) boils down to the question how many $y$ there are such that there is no $v$ with $r(x,y)=g(x,y)-v$. We want these to be at least $f(x)$, so we want at least $2^{p(|x|)} (2^{p(|x|)} - 1)+f(x)$ pairs such that $r(x,y)\ne g(x,y)- v$.
\end{proof}

\begin{lem}\cite{green93}\label{lem:green}
 $\exists \cne \p = \cne \p$.
\end{lem}

\begin{lem}\cite{schonhage79}\label{lem:schonhage}
For a large enough constant $c>0$, it holds that for any integers $n$ and $x$ with $|x| \leqslant 2^{2^n}$
and $x \neq 0$, the number of primes $p$ smaller than $2^{cn}$ such that $x \not\equiv 0 \mod p$
is at least $2^{cn}/cn$.
\end{lem}

\begin{lem}\cite[p. 81]{hemaspaandraO02}\label{lem:oraclePP}
 For every oracle $X$ we have $\PP^{\BPP^X} = \PP^X$.
\end{lem}

\subparagraph*{Arithmetic circuits}
An \emph{arithmetic circuit} is a labeled directed acyclic graph (DAG) consisting of vertices or gates with indegree or fanin $0$ or $2$. The gates with fanin $0$ are called input gates and are labeled with $-1$ or variables $X_1, X_2, \ldots, X_n$. The gates with fanin $2$ are called computation gates and are labeled with $\times$ or $+$. We can also consider circuits where computation gates may receive more than two edges, in which case we say that they have \emph{unbounded fanin}.
The polynomial computed by an arithmetic circuit is defined in the obvious way: an input gate computes the value of its label, a computation gate computes the product or the sum of its children's values, respectively. We assume that a circuit has only one sink which we call the output gate. We say that the polynomial computed by the circuit is the polynomial computed by the output gate. 
The \emph{size} of an arithmetic circuit is the number of gates. The \emph{depth} of a circuit is the length of the longest path from an input gate to the output gate in the circuit. A formula is an arithmetic circuit whose underlying graph is a tree. Finally, a circuit or formula is called \emph{monotone} if, instead of the constant $-1$, only the constant $1$ is allowed.

It is common to consider so-called \emph{degree-bounded} arithmetic circuits,  for which the degree of the computed polynomial  is bounded polynomially in the number of gates of the circuit. In our opinion this kind of degree bound has two problems. One is that computing the degree of a polynomial represented by a circuit is suspected to be hard (see~\cite{allenderBKM09,koiranP07,kayalS11}), so problems defined with this degree bound must often be promise problems. The other problem is that the bound on the degree does not bound the size of computed constants, which by iterative squaring can have exponential bitsize. Thus even evaluating circuits on a Turing machine becomes intractable. The paper by Allender et al.~\cite{allenderBKM09} discusses problems that result from this. To avoid all these complications, instead of bounding the degree of the computed polynomial, we choose 
to bound the formal degree of the circuit or equivalently 
to consider multiplicatively disjoint circuits. A circuit is called \emph{multiplicatively disjoint} if,  for each $\times$-gate, its two input   subcircuits are  disjoint from one another. See~\cite{malodP08} for a discussion of degree, formal degree and multiplicative disjointness and how they relate.

\section{Zero monomial coefficient}\label{sec:zmc}

We first consider the question of deciding if a single specified monomial occurs in a polynomial. In this problem and others regarding monomials, a monomial is encoded by giving the variable powers in binary.

\problem{ZMC}{Arithmetic circuit $C$, monomial $m$.}{Decide if $m$ has the coefficient $0$ in the polynomial computed by $C$.}

\begin{thm}\label{th:zmccep}
\pn{ZMC} is $\ce\p$-complete for both multiplicatively disjoint circuits and formulas.
\end{thm}

\begin{proof}
Using standard reduction techniques from the $\sP$-completeness of the permanent  (see for example~\cite{AB09}), one define the following generic  $\ce\p$-complete problem, as mentioned in the introduction.

\problem{$\pn{per}_{=}$}{Matrix $A\in \{0,1,-1\}^n$, $d\in \mathbb{N}$.}{Decide if $\pn{per}(A) = d$.}

Therefore, for the hardness of \pn{ZMC} it is sufficient to show a reduction from $\pn{per}_{=}$. 
 We use the following classical argument. On input $A=(a_{ij})$ and $d$ we compute the formula $Q:= \prod_{i=1}^n \left( \sum_{j=1}^n a_{ij} Y_j\right)$.
It is a classical observation by Valiant~\cite{valiant79}\footnote{According to~\cite{vzG87} this observation even goes back to~\cite{hammond79}.} that the monomial $Y_1Y_2\ldots Y_n$ has the coefficient $\pn{per}(A)$. Thus the coefficient of the monomial  $Y_1Y_2\ldots Y_n$ in $Q-d Y_1Y_2\ldots Y_n$ is $0$ if and only if $\pn{per}(A) = d$.

 We now show that \pn{ZMC} for multiplicatively disjoint circuits is in $\ce\p$. 
The proof is based on the use of parse trees, which can be seen as objects tracking the formation of monomials during the computation~\cite{malodP08} and are the algebraic analog of proof trees~\cite{vt89}.
A brief description is given in Appendix~\ref{app:parsetrees}.

  Consider a multiplicatively disjoint circuit $C$ and a monomial $m$, where the input gates of $C$ are labeled either by a variable or by $-1$. A parse tree $T$  contributes  to the monomial $m$ in the output polynomial if, when computing the value  of the tree, we get exactly the powers in $m$; this contribution has coefficient $+1$ if the number of gates labeled $-1$ in $T$ is even and it has coefficient $-1$ if this number is odd.
The coefficient of $m$ is thus equal to $0$ if and only if  the number of trees contributing positively is equal to the number of trees contributing negatively.

Let us represent a parse tree by a boolean word $\bar{\lep}$, by indicating which edges of $C$ appear in the parse tree (the length $N$ of the words is therefore the number of edges in $C$).  Some of these words will not represent a valid parse tree, but this can be tested in polynomial time.
Consider the following language $L$ composed of triples $(C,m,\lep_0\bar{\lep})$ such that:
\begin{enumerate}
\item $\lep_0=0$ and $\bar{\lep}$ encodes a valid parse tree of $C$ which contribute positively to $m$,
\item or $\lep_0=1$ and   $\bar{\lep}$ does not encode a valid parse tree contributing negatively to $m$.
\end{enumerate}
Then the number of $\bar{\lep}$ such that $(C,m,0\bar{\lep})$ belongs to $L$ is the number of parse trees contributing positively to $m$ and the number of $\bar{\lep}$ such that $(C,m,1\bar{\lep})$ belongs to $L$ is equal to $2^N$ minus the number of parse trees contributing negatively to $m$. Thus, the number of $\lep_0\bar{\lep}$ such that $(C,m,\lep_0\bar{\lep})\in L$ is equal to $2^N$ if and only if  the number of trees contributing positively is equal to the number of trees contributing negatively, if and only if  
the coefficient of $m$ is  equal to $0$ in $C$. Because $L$ is in $\p$, \pn{ZMC}  for  multiplicatively disjoint circuits is in $\ce\p$.
\end{proof}

\begin{thm}\label{lem:zmcupper}
\pn{ZMC} belongs to $\coRP^\PP$.
\end{thm}

\begin{proof}
Given a circuit $C$, a monomial $m$ and a prime number $p$ written in binary,
 \pn{CoeffSLP} is the problem of computing modulo $p$ the coefficient of the monomial $m$ 
in the polynomial computed by $C$. It is shown in~\cite{kayalS11}
that \pn{CoeffSLP} belongs to $\FP^\sP$.

We now describe a randomized algorithm to decide \pn{ZMC}.
Let $c$ be the constant given in Lemma~\ref{lem:schonhage}.
Consider the following algorithm to decide $\pn{ZMC}$ given a circuit $C$ of size $n$
and a monomial $m$, using $\pn{CoeffSLP}$ as an oracle. First choose uniformly at random
an integer $p$ smaller than $2^{cn}$.
If $p$ is not prime, accept. Otherwise, compute the coefficient $a$
of the monomial $m$ in $C$ with the help of the oracle and accept if $a \equiv 0 \mod p$.
Since $|a| \leq 2^{2^n}$, Lemma~\ref{lem:schonhage} ensures that the above is a correct
one-sided error probabilistic algorithm for \pn{ZMC}. 
This yields $\pn{ZMC} \in \coRP^\pn{CoeffSLP}$. Hence $\pn{ZMC} \in \coRP^\PP$.
\end{proof}

\begin{thm}\label{th:zmc-monotone}
\pn{ZMC} is $\coNP$-complete both for monotone formulas
and monotone circuits.
\end{thm}
\begin{proof}
For hardness, we reduce the $\NP$-complete problem \pn{Exact-3-Cover}~\cite{gareyJ79} to
the complement of \pn{ZMC} on monotone formulas, as done in~\cite[Chapter 3]{strozecki:phd} (we reproduce the argument here for completeness).

\problem{Exact-3-Cover}{Integer $n$ and $C_1,\ldots,C_m$ some $3$-subsets
of $\{1,\ldots,n\}$.}{Decide if there exists $I \subseteq \{1,\ldots,m\}$
such that $\{C_i\ |\ i \in I\}$ is a partition of $\{1,\ldots,n\}$.}

Consider the formula $F=\prod_{i=1}^m (1+\prod_{j \in C_i} X_j)$. 
The monotone
formula $F$ has the monomial $\prod_{i=1}^n X_i$ if and only if $(n,C_1,\ldots,C_m)$
is a positive instance of \pn{Exact-3-Cover}.
 
 Let us now show that \pn{ZMC} for monotone circuits is in $\coNP$. This proof will use the notion of \emph{parse tree types}, which are inspired by the generic polynomial introduced in~\cite{mal07} to compute coefficient functions. We give here a sketch of the argument, more details are provided in
Appendix~\ref{app:parsetrees}.
The parse trees of a circuit which is not necessarily multiplicatively disjoint may be of a much bigger size than the circuit itself, because they can be seen as parse trees of the formula associated to the circuit and obtained by duplicating gates and edges. Define the \emph{type} of a parse tree by giving, for each edge in the original circuit, the number of copies of this edge in the parse tree. 
 There can be many different parse trees for a given parse tree type
 but
they will all contribute to the same monomial, which is easy to obtain from the type: the power of a variable in the monomial is  the sum, taken over all input gates labeled by this variable, of the number of edges leaving from this gate. In the case of a monotone circuit, computing the exact number of parse trees for a given type is thus not necessary, as a monomial will have a non-zero coefficient if and only if  there exists a valid parse tree type producing this monomial. 

 Parse tree types, much like parse trees in the proof of Theorem~\ref{th:zmccep}, can be represented by Boolean tuples which must satisfy some easy-to-check conditions to be valid. Thus the coefficient of a monomial is $0$ if and only if  there are no valid parse tree types producing this monomial, which is a $\coNP$ condition.
\end{proof}

\section{Counting monomials}\label{sec:monslp}

We now turn to the problem of counting the monomials of a polynomial represented by a circuit.

\problem{CountMon}{Arithmetic circuit $C$, $d\in \mathbb{N}$.}{Decide if the polynomial computed by $C$ has at least $d$ monomials.}

To study the complexity of \pn{CountMon} we will look at what we call extending polynomials.
Given two monomials $M$ and $m$, we say that $M$ is $m$-extending if $M = m m'$ and $m$ and $m'$ have no common variable. We start by studying the problem of deciding the existence of an extending monomial.

\problem{ExistExtMon}{Arithmetic circuit $C$, monomial $m$.}{Decide if the polynomial computed by $C$  contains an $m$-extending monomial.}

\begin{prop}\label{prop:existextending}
$\pn{ExistExtMon}$ is in $\RP^\PP$. For multiplicatively disjoint circuits it is $\cne \p$-complete.
\end{prop}
\begin{proof}
 We first show the first upper bound. So let $(C, m)$ be an input for $\pn{ExistExtMon}$ where $C$ is a circuit in the variables $X_1,\ldots, X_n$. Without loss of generality, suppose that $X_1,\ldots, X_r$ are the variables appearing in $m$. Let  $d=2^{|C|}$: $d$ is a bound on the degree of the polynomial computed by $C$. We define 
  $C' = \prod_{i=r+1}^n (1+ Y_i X_i)^d$ for new variables $Y_i$. 
We have that $C$ has an $m$-extending monomial if and only if in the product $CC'$ the  polynomial $P(Y_{r+1}, \ldots, Y_n)$, which is the coefficient of $m\prod_{i=r+1}^n X_i^d$, is not identically $0$. Observe that $P$ is not given explicitly but can be evaluated modulo a random prime with an oracle for $\pn{CoeffSLP}$. Thus it can be checked  if $P$ is identically $0$ with the classical Schwartz-Zippel-DeMillo-Lipton lemma (see for example~\cite{AB09}).
It follows that $\pn{ExistExtMon}\in \RP^\PP$.

The upper bound in the multiplicatively disjoint setting is easier: we can guess an $m$-extending monomial $M$ and then output the answer of an oracle for the complement of $\pn{ZMC}$, to check whether $M$ appears in the computed polynomial. This establishes containment in $\exists \cne \p$ which by Lemma \ref{lem:green} is $\cne \p$.

For hardness we reduce to $\pn{ExistExtMon}$ the $\cne\p$-complete problem $\pn{per}_{\ne}$,  i.e., the complement of the $\pn{per}_{=}$ problem introduced for the proof of Theorem~\ref{th:zmccep}.
We use essentially the same reduction constructing a circuit $Q:= \prod_{i=1}^n \left( \sum_{j=1}^n a_{ij} Y_j\right)$.
Observe that the only potential extension of $m:=Y_1Y_2\ldots Y_n$ is $m$ itself and has the coefficient $\pn{per}(A)$. Thus $Q-d Y_1Y_2\ldots Y_n$ has an $m$-extension if and only if $\pn{per}(A) \ne d$.
\end{proof}

\problem{CountExtMon}{Arithmetic circuit $C$, $d\in \mathbb{N}$, monomial $m$.}{Decide if the polynomial computed by $C$ has at least $d$ $m$-extending monomials.}

\begin{prop}\label{prop:toextending}
 $\pn{CountExtMon}$ is $\PP^\PP$-complete.
\end{prop}
\begin{proof}
Clearly $\pn{CountExtMon}$ belongs to $\PP^{\pn{ZMC}}$ and thus with Theorem~\ref{lem:zmcupper} it is in $\PP^{\coRP^\PP}$. Using Lemma~\ref{lem:oraclePP} we get  membership in $ \PP^\PP$.
To show hardness, we reduce the canonical $\cc\cne \p$-complete problem $\cc\cne\pn{3SAT}$ to $\pn{CountExtMon}$. With Lemma~\ref{lem:oraclesPPPP} the hardness for $\PP^\PP$ follows.

\problem{$\cc\cne$3SAT}{$\pn{3SAT}$-formula $F(\bar{x}, \bar{y})$, $k,\ell \in \mathbb{N}$.}{Decide if there are at least $k$ assignments to $\bar{x}$ such that there are not exactly $\ell$ assignments to $\bar{y}$ such that $F$ is satisfied.}

Let $(F(\bar{x}, \bar{y}), k,\ell)$ be an instance for $\cc\cne\pn{3SAT}$. Without loss of generality we may assume that $\bar{x}= (x_1, \ldots, x_n)$ and $\bar{y} = (y_1, \ldots, y_n)$ and that no clause contains a variable in both negated and unnegated form. Let $\Gamma_1, \ldots, \Gamma_c$ be the clauses of $F$.

For each literal $u$ of the variables in $\bar{x}$ and $\bar{y}$ we define a monomial $I(u)$ in the variables $X_1, \ldots , X_n, Z_1, \ldots , Z_c$ in the following way:
 \begin{align*}
  I(x_i) &= X_i \prod_{\{j\ |\ x_i\in \Gamma_j\}} Z_j, &
 I(\lnot x_i) &= \prod_{\{j\ |\ \lnot x_i\in \Gamma_j\}} Z_j,\\[.4cm]
  I(y_i) &= \prod_{\{j\ |\ y_i\in \Gamma_j\}} Z_j, &
 I(\lnot y_i) &= \prod_{\{j\ |\ \lnot y_i\in \Gamma_j\}} Z_j.
 \end{align*}
From these monomials we compute a formula $C$ by
 \begin{equation}\label{eq:polynomial-from-3sat}
 C:= \prod_{i=1}^n \left(I(x_i)+I(\lnot x_i)\right)\prod_{i=1}^n \left(I(y_i)+I(\lnot y_i)\right).\end{equation}

 We fix a mapping $mon$ from the assignments of $F$ to the monomials computed by $C$: Let $\bar{\alpha}$ be an assignment to $\bar{x}$ and $\bar{\beta}$ be an assignment to $\bar{y}$. We define  $mon(\bar{\alpha}\bar{\beta})$ 
 as the monomial obtained in the expansion of $C$  by choosing the following terms. If $\alpha_i = 0$, choose $I(\lnot x_i)$, otherwise choose $I(x_i)$. Similarly, if $\beta_i = 0$, choose $I(\lnot y_i)$, otherwise choose $I(y_i)$. 

The monomial $mon(\bar{\alpha}\bar{\beta})$ has the form $\prod_{i=1}^n X_i^{\alpha_i} \prod_{j=1}^c Z_j^{\gamma_j}$, where $\gamma_j$ is the number of true literals in $\Gamma_j$ under the assignment $\bar{\alpha}\bar{\beta}$. Then $F$ is true under $\bar{\alpha}\bar{\beta}$ if and only if $mon(\bar{\alpha}\bar{\beta})$ has the factor $\prod_{j=1}^c Z_j$. 
Thus $F$ is true under $\bar{\alpha}\bar{\beta}$ if and only if  $mon(\bar{\alpha}\bar{\beta}) \prod_{j=1}^c \left(1+ Z_j + Z_j^2\right)$ has the factor $\prod_{i=1}^n X_i^{\alpha_i} \prod_{j=1}^c Z_j^3$. We set $C'= C \prod_{j=1}^c \left(1+ Z_j + Z_j^2\right)$.
 
Consider an assignment $\bar{\alpha}$ to $\bar{x}$.
 The coefficient of the monomial $\prod_{i=1}^n X_i^{\alpha_i} \prod_{j=1}^c Z_j^3$ in $C'$ is the number of assignments $\bar{\beta}$ such that $\bar{\alpha}\bar{\beta}$  satisfies $F$. Thus we get
\begin{eqnarray*}
 &&(F(\bar{x}, \bar{y}), k,\ell)\in \cc\cne\pn{3SAT}\\
&\Leftrightarrow &\text{there are at least $k$ assignments $\bar{\alpha}$  to $\bar{x}$ such that the monomial $\prod_{i=1}^n X_i^{\alpha_i} \prod_{j=1}^c Z_j^3$} \\&&\text{does not have  coefficient $\ell$ in $C'$}\\
&\Leftrightarrow & \text{there are at least $k$ assignments $\bar{\alpha}$  to $\bar{x}$ such that the monomial $\prod_{i=1}^n X_i^{\alpha_i} \prod_{j=1}^c Z_j^3$} \\&&\text{occurs in $C'' := C'- \ell\prod_{i=1}^n(1+X_i)\prod_{j=1}^c Z_j^3$}\\
&\Leftrightarrow & \text{there are at least $k$ tuples $\bar{\alpha}$ such that  $C''$ contains the monomial }   \prod_{i=1}^n X_i^{\alpha_i} \prod_{j=1}^c Z_j^3\\
&\Leftrightarrow &  C'' \text{ has at least $k$ $(\prod_{j=1}^c Z_j^3)$-extending monomials}.
\end{eqnarray*}

\end{proof}

\begin{thm}\label{thm:monSLP}
\pn{CountMon} is $\PP^\PP$-complete.
It is $\PP^\PP$-hard even for unbounded fan-in formulas of depth~$4$.
\end{thm}

\begin{proof} \pn{CountMon} can be easily reduced to \pn{CountExtMon} since the number of monomials of a polynomial is the number of $1$-extending monomials. Therefore \pn{CountMon} belongs to $\PP^\PP$.

 To show hardness, 
it is enough to prove that instances of $\pn{CountExtMon}$ constructed in Proposition~\ref{prop:toextending} can be reduced to  $\pn{CountMon}$ in logarithmic space.
The idea of the proof is that we make sure that the polynomial for which we count all monomials contains all monomials that are not $m$-extending. Thus we know how many non-$m$-extending monomials it contains and we can compute the number of $m$-extending monomials from the number of all monomials. We could use the same strategy to  show in general that \pn{CountExtMon} reduces to \pn{CountMon} but by considering the instance obtained in the proof of Proposition~\ref{prop:toextending} and analyzing the extra calculations below we get hardness for unbounded fanin formulas of depth~$4$.

So let $(C'', k, m)$ be the instance of $\pn{CountExtMon}$ constructed in the proof of Proposition~\ref{prop:toextending}, with  $m=\prod_{j=1}^c Z_j^3$.   We therefore  need to count the monomials computed by $C''$ which are of the form $f(X_1, \ldots, X_n)\prod_{j=1}^c Z_j^3$. 
The circuit $C''$ is multilinear in $X$, and the $Z_j$ can only appear with powers in $\{0,1,2,3,4,5\}$. So the non-$m$-extending monomials computed by $C''$ are all products of a multilinear monomial in the $X_i$ and a monomial in the $Z_j$ where at least one $Z_j$ has a power in $\{0,1,2,4,5\}$. Fix $j$, then all monomials that are not $m$-extending because of $Z_j$ are computed by the formula
\begin{equation}\label{eq:extra-monomials}
\tilde{C_j}:= \left(\prod_{i=1}^n (X_i + 1)\right) \left(\prod_{j'\ne j} \sum_{p=0}^5 Z_{j'}^p \right) \left(1+Z_j+Z_j^2+Z_j^4+Z_j^5\right).
\end{equation}

Thus the formula $\tilde{C}:= \sum_j \tilde{C_j}$ computes all non-$m$-extending monomials that $C''$ can compute. The coefficients of monomials in $C''$ cannot be smaller than $-\ell$ where $\ell$ is part of the instance of $\cc\cne$\pn{3SAT} from which we constructed $(C'',k,m)$ before. So the formula $C^*:=C''+(\ell+1)\tilde{C}$ contains all non-$m$-extending monomials that $C''$ can compute and it contains the same extending monomials. There are $2^n6^c$ monomials of the form that $C''$ can compute, only $2^n$ of which are $m$-extending, which means that there are $2^n(6^c-1)$ monomials computed by $C^*$ that are not $m$-extending. As a consequence, $C''$ has at least $k$ $m$-extending monomials if and only if $C^*$ has at least $2^n(6^c-1)+k$ monomials.
\end{proof}

\begin{thm}\label{thm:monmonSLP}
\pn{CountMon} is $\PP^\NP$-complete
both for monotone formulas and monotone circuits.
\end{thm}

\begin{proof}
We first show hardness for monotone formulas.
The argument is very similar to the proof of Theorem~\ref{thm:monSLP}. 
Consider the following canonical $\cc\exists\p$-complete problem $\cc\exists\pn{3SAT}$.

\problem{$\cc\exists$3SAT}{$\pn{3SAT}$-formula $F(\bar{x}, \bar{y})$, $k \in \mathbb{N}$.}
{Decide if there are at least $k$ assignments $\bar{\alpha}$ to $\bar{x}$ such that
$F(\bar{\alpha}, \bar{y})$ is satisfiable.}

We reduce $\cc\exists\pn{3SAT}$ to \pn{CountMon}. With Lemma~\ref{lem:quantorsvsoracles}
the hardness for $\PP^\NP$ follows.
Consider a $\pn{3SAT}$-formula $F(\bar{x}, \bar{y})$.
Let $n=|\bar{x}|=|\bar{y}|$ and let $c$ be the number of clauses of $F$.
Define the polynomial
$C^* = C+\sum_{j=1}^c \tilde{C_j}$
where $C$ is defined by Equation~\ref{eq:polynomial-from-3sat} and $\tilde{C_j}$
by Equation~\ref{eq:extra-monomials}.
The analysis is similar to the proof of Theorem~\ref{thm:monSLP}.
The polynomial $C^*$ is computed by a monotone arithmetic formula
and has at least $2^n (6^c-1)+k$ monomials if and only if $(F,k)$ is a positive
instance of $\cc\exists\pn{3SAT}$.

We now prove the upper bound. Recall that $\pn{CountMon} \in \PP^{\pn{ZMC}}$.
From Theorem~\ref{th:zmc-monotone}, 
it follows that $\pn{CountMon}$ on monotone circuits belongs to $\PP^{\NP}$.
\end{proof}

\section{Multilinearity}\label{sec:ml}

In this section we consider the effect of multilinearity on our problems. We will not consider promise problems and therefore  the multilinear variants of our problems must first check if the computed polynomial is multilinear. We start by showing that this step is not difficult.

\problem{CheckML}{Arithmetic circuit $C$.}{Decide if the polynomial computed by $C$ is multilinear.}

\begin{prop}\label{prop:checkml} \pn{CheckML} is equivalent to \pn{ACIT}.
\end{prop}
\begin{proof} Reducing \pn{ACIT} to \pn{CheckML} is easy: Simply multiply the input with $X^2$ for an arbitrary variable $X$. The
 resulting circuit is multilinear if and only if the original circuit was $0$.

For the other direction the idea is to compute the second derivatives of the polynomial computed by the input circuit and check if they are $0$.

So let $C$ be a circuit in the variables $X_1, \ldots, X_n$ that is to be checked for multilinearity. For each $i$ we inductively compute a circuit $C_i$ that computes the second derivative with respect to $X_i$. To do so for each gate $v$ in $C$ the circuit $C'$ has three gates $v_i$, $v_i'$ and $v_i''$. 
The polynomial in $v_i$ is that of $v$, $v_i'$ computes the first derivative and $v_i''$ the second. For the input gates the construction is obvious. If $v$ is a $+$-gate with children $u$ and $w$ we have $v_i= u_i+w_i$, $v_i'= u_i'+w_i'$ and $v_i''= u_i''+w_i''$. If $v$ is a $\times$-gate with children $u$ and $w$ we have $v_i= u_i w_i$, $v_i'= u_i' w_i +u_i w_i'$ and $v_i''= u_i'' w_i + 2 u_i' w_i'+ u_i w_i''$. It is easy to see that the constructed circuit computes indeed the second derivative with respect to $X_i$.

Next we compute $C':= \sum_{i=1}^n Y_i C_i$ for new variables $Y_i$. We have that $C'$ is identically zero if and only if $C$ is multilinear. Also $C'$ can easily be constructed in logarithmic space.
\end{proof}

Next we show that the problem gets much harder if, instead of asking whether \emph{all} the monomials in the polynomial  computed by a circuit are multilinear,  we ask whether at least \emph{one} of the monomials is multilinear. 

\problem{MonML}{Arithmetic circuit $C$.}{Decide if the polynomial computed by $C$ contains a multilinear monomial.}

The problem \pn{monML} lies at the heart of fast exact algorithms for deciding $k$-paths by Koutis and Williams~\cite{koutis08,williams09} (although in these papers the polynomials are in characteristic $2$ which changes the problem a little). This motivated Chen and Fu \cite{CF10, CF11} to consider \pn{monML}, show that it is $\sP$-hard and give algorithms for the bounded depth version.
We provide further information on the complexity of this problem.

\begin{prop}\label{prop:monml}
\pn{MonML} is in $\RP^\PP$. It is  $\cne\p$-complete for multiplicatively disjoint circuits.\end{prop}

\begin{proof}
For the first upper bound, let $C$ be the input in variables $X_1,\ldots, X_n$. We set $C'= \prod_{i=1}^n (1+ X_i Y_i)$.
Then $C$ computes a multilinear monomial if and only if in the product $CC'$ the coefficient polynomial $P(Y_1, \ldots, Y_n)$ of $\prod_{i=1}^n X_i$ is not identically $0$. This can be tested as in the proof of Proposition~\ref{prop:existextending},  thus establishing $\pn{MonML}\in \RP^\PP$.
 
The $\cne\p$-completeness in the multiplicatively disjoint case can be proved in the same way as in Proposition~\ref{prop:existextending}.
\end{proof}

We now turn to our first problem, namely deciding whether a monomial appears in the polynomial computed by a circuit, in the multilinear setting.

\problem{ML-ZMC}{Arithmetic circuit $C$, monomial $m$.}{Decide if $C$ computes a multilinear polynomial in which the monomial $m$ has  coefficient $0$.}

\begin{prop}\label{prop:zmcml} \pn{ML-ZMC} is equivalent to \pn{ACIT}.
\end{prop}

\begin{proof}
 We first show that \pn{ACIT} reduces to \pn{ML-ZMC}. So let $C$ be an input for \pn{ACIT}. Allender et al.~\cite{allenderBKM09} have shown that \pn{ACIT} reduces to a restricted version of \pn{ACIT} in which all inputs are $-1$ and thus the circuit computes a constant. Let $C_1$ be the result of this reduction. 
Then $C$ computes identically $0$ if and only if the constant coefficient of $C_1$ is $0$. This establishes the first direction.

For the other direction let $(C,m)$ be the input, where $C$ is an arithmetic circuit and $m$ is a monomial. 
First check if $m$ is multilinear, if not output $1$ or any other nonzero polynomial. 
Next we construct a circuit $C_1$ that computes the homogeneous component of degree $\deg(m)$ of $C$ with the classical method  (see for example~\cite[Lemma 2.14]{bur00}).
Observe that if $C$ computes a multilinear polynomial, so does $C_1$. We now plug in $1$ for the variables that appear in $m$ and $0$ for all other variables, call the resulting (constant) circuit $C_2$. If $C_1$ computes a multilinear polynomial, then $C_2$ is zero if and only if $m$ has coefficient $0$ in $C_1$. The end result of the reduction is $C^*:=   C_2+Z C_3$ where $Z$ is a new variable and $C_3$ is a circuit which is identically $0$ iff $C$ computes a multilinear polynomial (obtained via Proposition~\ref{prop:checkml}).
$C$ computes a multilinear polynomial and does not contain the monomial $m$ if and only if both  $C_2$ and $Z C_3$   are identically $0$, which happens if and only if their sum is identically $0$.
\end{proof}

In the case of our second problem, counting the number of monomials, the complexity falls to $\mathsf{PP}$.

\problem{ML-CountMon}{Arithmetic circuit $C$, $d\in \mathbb{N}$.}{Decide if the polynomial computed by $C$ is multilinear and has at least $d$ monomials.}

\begin{prop}\label{prop:monslpml} \pn{ML-CountMon} is $\mathsf{PP}$-complete (for Turing reductions).\end{prop}

\begin{proof}
We first show $\pn{ML-CountMon}\in \mathsf{PP}$. To do so we  use $\pn{CheckML}$ to check that the polynomial computed by $C$ is multilinear. Then  counting monomials can  be done in $\PP^\pn{ML-ZMC}$, and  $\pn{ML-ZMC}$ is in $\coRP$. By Lemma \ref{lem:oraclePP} the  class $\PP^{\coRP}$ is simply $\PP$.

For hardness we reduce the computation of the  $\{0,1\}$-permanent to \pn{ML-CountMon}. The proposition follows, because 
the  $\{0,1\}$-permanent is $\sP$-complete for Turing reductions.
So let $A$ be a $0$-$1$-matrix and $d\in \mathbb{N}$ and we have to decide if $\pn{per}(A)\ge d$. We get a matrix $B$ from $A$ by setting $b_{ij} := a_{ij} X_{ij}$. Because every entry of $B$ is either $0$ or a distinct variable,
we have that, when we compute the permanent of $B$, every permutation that yields a non-zero summand yields a unique monomial. This means that there are no cancellations, so that $\pn{per}(A)$ is the number of monomials in $\pn{per}(B)$. 

The problem is now  that no small circuits for the permanent are known and thus $\pn{per}(B)$ is not a good input for $\pn{ML-CountMon}$. But because there are no cancellations, we have that $\pn{det}(B)$ and $\pn{per}(B)$ have the same number of monomials. So take a small circuit for the determinant (for instance the one given in~\cite{mahajanV97}) and substitute its inputs by the entries of $B$. The result is a circuit $C$ which computes a polynomial whose number of monomials is $\pn{per}(A)$. Observing that the determinant, and thus the polynomial computed by $C$, is multilinear completes the proof.
\end{proof}

\subparagraph*{Acknowledgements}

We would like to thank Sylvain Perifel for helpful discussions. The results of this paper were conceived while the third author was visiting the Équipe de Logique Mathématique at Université Paris Diderot Paris 7. He would like to thank Arnaud Durand for making this stay possible, thanks to funding from ANR ENUM (ANR-07-BLAN-0327). The third author would also like to thank his supervisor Peter Bürgisser for helpful advice.

\bibliographystyle{plain}
\bibliography{literature,crossref}

\begin{thebibliography}{10}

\bibitem{ABO99}
E.~Allender, R.~Beals, and M.~Ogihara.
\newblock The complexity of matrix rank and feasible systems of linear
  equations.
\newblock {\em Computational Complexity}, 8(2):99--126, 1999.

\bibitem{allenderBKM09}
E.~Allender, P.~B{\"u}rgisser, J.~Kjeldgaard-Pedersen, and P.~B. Miltersen.
\newblock {On the Complexity of Numerical Analysis}.
\newblock {\em SIAM J. Comput.}, 38(5):1987--2006, 2009.

\bibitem{AW93}
E.W. Allender and K.W. Wagner.
\newblock Counting hierarchies: polynomial time and constant depth circuits.
\newblock {\em Current trends in theoretical computer science: essays and
  Tutorials}, 40:469, 1993.

\bibitem{AB09}
S.~Arora and B.~Barak.
\newblock {\em Computational complexity: a modern approach}.
\newblock Cambridge University Press, 2009.

\bibitem{bur00}
P.~B{\"u}rgisser.
\newblock {\em Completeness and reduction in algebraic complexity theory},
  volume~7.
\newblock Springer Verlag, 2000.

\bibitem{Burgisser09}
P.~B{\"u}rgisser.
\newblock {On Defining Integers And Proving Arithmetic Circuit Lower Bounds}.
\newblock {\em Computational Complexity}, 18(1):81--103, 2009.

\bibitem{CF10}
Zhixiang Chen and Bin Fu.
\newblock Approximating multilinear monomial coefficients and maximum
  multilinear monomials in multivariate polynomials.
\newblock In Weili Wu and Ovidiu Daescu, editors, {\em COCOA (1)}, volume 6508
  of {\em Lecture Notes in Computer Science}, pages 309--323. Springer, 2010.

\bibitem{CF11}
Zhixiang Chen and Bin Fu.
\newblock {The Complexity of Testing Monomials in Multivariate Polynomials}.
\newblock In Weifan Wang, Xuding Zhu, and Ding-Zhu Du, editors, {\em COCOA},
  volume 6831 of {\em Lecture Notes in Computer Science}, pages 1--15.
  Springer, 2011.

\bibitem{gareyJ79}
M.~R. Garey and D.~S. Johnson.
\newblock {\em {Computers and Intractability: A Guide to the Theory of
  NP-Completeness}}.
\newblock W. H. Freeman, 1979.

\bibitem{green93}
F.~Green.
\newblock {On the Power of Deterministic Reductions to $C_=P$}.
\newblock {\em Theory of Computing Systems}, 26(2):215--233, 1993.

\bibitem{hammond79}
J.~Hammond.
\newblock Question 6001.
\newblock {\em Educ. Times}, 32:179, 1879.

\bibitem{hemaspaandraO02}
L.A. Hemaspaandra and M.~Ogihara.
\newblock {\em The complexity theory companion}.
\newblock Springer Verlag, 2002.

\bibitem{JS11}
M.~Jansen and R.~Santhanam.
\newblock {Permanent Does Not Have Succinct Polynomial Size Arithmetic Circuits
  of Constant Depth}.
\newblock In Luca Aceto, Monika Henzinger, and Jiri Sgall, editors, {\em
  ICALP}, volume 6755 of {\em Lecture Notes in Computer Science}, pages
  724--735. Springer, 2011.

\bibitem{ki04}
V.~Kabanets and R.~Impagliazzo.
\newblock {Derandomizing Polynomial Identity Tests Means Proving Circuit Lower
  Bounds}.
\newblock {\em Computational Complexity}, 13:1--46, 2004.

\bibitem{kayalS11}
N.~Kayal and C.~Saha.
\newblock {On the Sum of Square Roots of Polynomials and related problems}.
\newblock In {\em IEEE Conference on Computational Complexity}, pages 292--299,
  2011.

\bibitem{koiranP07}
P.~Koiran and S.~Perifel.
\newblock {The complexity of two problems on arithmetic circuits}.
\newblock {\em Theor. Comput. Sci.}, 389(1-2):172--181, 2007.

\bibitem{KP11}
P.~Koiran and S.~Perifel.
\newblock {Interpolation in Valiant’s Theory}.
\newblock {\em Computational Complexity}, pages 1--20, 2011.

\bibitem{koutis08}
I.~Koutis.
\newblock {Faster Algebraic Algorithms for Path and Packing Problems}.
\newblock In Luca Aceto, Ivan Damg{\aa}rd, Leslie~Ann Goldberg, Magn{\'u}s~M.
  Halld{\'o}rsson, Anna Ing{\'o}lfsd{\'o}ttir, and Igor Walukiewicz, editors,
  {\em ICALP}, volume 5125 of {\em Lecture Notes in Computer Science}, pages
  575--586. Springer, 2008.

\bibitem{kbvdg11}
J.~H.~P. Kwisthout, H.~L. Bodlaender, and L.~C. Van Der~Gaag.
\newblock The complexity of finding kth most probable explanations in
  probabilistic networks.
\newblock In {\em Proceedings of the 37th international conference on Current
  trends in theory and practice of computer science}, SOFSEM'11, pages
  356--367, Berlin, Heidelberg, 2011. Springer-Verlag.

\bibitem{mahajanV97}
M.~Mahajan and V.~Vinay.
\newblock A combinatorial algorithm for the determinant.
\newblock In {\em Proceedings of the eighth annual ACM-SIAM symposium on
  Discrete algorithms}, pages 730--738. Society for Industrial and Applied
  Mathematics, 1997.

\bibitem{mal07}
G.~Malod.
\newblock {The Complexity of Polynomials and Their Coefficient Functions}.
\newblock In {\em IEEE Conference on Computational Complexity}, pages 193--204.
  IEEE Computer Society, 2007.

\bibitem{malodP08}
G.~Malod and N.~Portier.
\newblock {Characterizing Valiant's algebraic complexity classes}.
\newblock {\em J. Complexity}, 24(1):16--38, 2008.

\bibitem{mundhenk00}
M.~Mundhenk, J.~Goldsmith, C.~Lusena, and E.~Allender.
\newblock {Complexity of Finite-Horizon Markov Decision Process Problems}.
\newblock {\em Journal of the ACM (JACM)}, 47(4):681--720, 2000.

\bibitem{schonhage79}
A.~Sch{\"o}nhage.
\newblock {On the Power of Random Access Machines}.
\newblock In Hermann~A. Maurer, editor, {\em ICALP}, volume~71 of {\em Lecture
  Notes in Computer Science}, pages 520--529. Springer, 1979.

\bibitem{strozecki:phd}
Y.~Strozecki.
\newblock {\em Enumeration complexity and matroid decomposition}.
\newblock PhD thesis, Universit\'e Paris Diderot - Paris 7, 2010.

\bibitem{toran88}
J.~Tor{\'a}n.
\newblock {Succinct Representations of Counting Problems}.
\newblock In Teo Mora, editor, {\em AAECC}, volume 357 of {\em Lecture Notes in
  Computer Science}, pages 415--426. Springer, 1988.

\bibitem{toran91}
J.~Tor{\'a}n.
\newblock {Complexity Classes Defined by Counting Quantifiers}.
\newblock {\em J. ACM}, 38(3):753--774, 1991.

\bibitem{valiant79}
L.G. Valiant.
\newblock Completeness classes in algebra.
\newblock In {\em Proceedings of the eleventh annual ACM symposium on Theory of
  computing}, pages 249--261. ACM, 1979.

\bibitem{vt89}
H.~Venkateswaran and M.~Tompa.
\newblock A {N}ew {P}ebble {G}ame that {C}haracterizes {P}arallel {C}omplexity
  {C}lasses.
\newblock {\em SIAM J. Comput.}, 18(3):533--549, 1989.

\bibitem{vzG87}
J.~Von Zur~Gathen.
\newblock {Feasible arithmetic computations: Valiant's hypothesis}.
\newblock {\em Journal of Symbolic Computation}, 4(2):137--172, 1987.

\bibitem{wagner86}
K.~W. Wagner.
\newblock {The Complexity of Combinatorial Problems with Succinct Input
  Representation}.
\newblock {\em Acta Informatica}, 23(3):325--356, 1986.

\bibitem{williams09}
R.~Williams.
\newblock {Finding paths of length $k$ in $O^*(2^k)$ time}.
\newblock {\em Information Processing Letters}, 109(6):315--318, 2009.

\end{thebibliography}

\appendix
 
\section{Parse trees and coefficients}\label{app:parsetrees}
  
Define inductively the parse trees of a circuit $C$ in the following manner:
\begin{enumerate}
\item the only parse tree of an input gate is the gate itself,
\item the  parse trees of an addition gate \lal\ with argument gates \lbe\ and \lga\ are obtained by taking either  a parse
tree of \lbe\ and adding the edge from \lbe\ to \lal\ or by taking a parse tree of \lga\ and adding the edge from \lga\ to \lal,
\item the  parse trees of a multiplication gate \lal\ with argument gates \lbe\ and \lga\ are obtained by taking  a parse
tree of \lbe\ and a parse tree of \lga\ and adding the edge from \lbe\ to \lal\  and  the edge from \lga\ to \lal, renaming vertices so that the chosen parse trees of \lbe\ and \lga\ are disjoint.
\end{enumerate}
The value of a parse tree is defined as the product of the labels of each input gate in the parse tree (note that in the parse tree there may be several copies of
a given input gate of the circuit, so that the corresponding label will have as power the number of copies of the gate).
It is easy to see that the polynomial computed by a circuit is the sum of the values of its parse trees:
$$C(\bar{x})=\sum_{T\text{ parse tree of }C} \text{value}(T).$$
 
In the case of a multiplicatively disjoint circuit,  any parse tree is a subgraph of the circuit. 
In this case,  a parse tree can be equivalently seen as a subgraph  defined by a subset of $T$ of the edges satisfying the following properties:
\begin{enumerate}
\item it contains the output gate,
\item for any addition gate \lal, if $T$ contains an edge with origin \lal, then $T$ contains exactly one edge with destination \lal,
\item for any multiplication gate \lal, if $T$ contains an edge with origin \lal, then $T$ contains both (all) edges with destination \lal,
\item for any gate \lal, if $T$ contains an edge with destination \lal, then $T$ contains an edge with origin \lal.
\end{enumerate}
 
In the case of an arbitrary circuit, the set of parse trees of a circuit $C$ is equal to the set of parse trees of the associated arithmetic formula (or of the associated multiplicatively disjoint circuit). But a parse tree of a circuit $C$ need not be a subgraph, in particular the size of a parse tree may be exponential in the size of the circuit, because of the duplication involved in the definition of parse trees. This means that for a given edge $e$ in the circuit $C$, there may be several copies of $e$ in a parse tree. We will call the number of copies the \emph{multiplicity} of $e$.
 
The type \ltau\ of  a parse tree $T$ for a circuit $C$ is the function which associates to each edge $e$ of $C$ its multiplicity $\ltau(e)$ in $T$.
One could give an inductive characterization of parse tree types similar to the definition of parse trees above, but the following characterization is more useful. Add a unique artificial edge $e_{\text{out}}$ whose origin is the output gate of the circuit.
A parse tree type \ltau\ must satisfy the  properties below:
\begin{enumerate}
\item the multiplicity of $e_{\text{out}}$ is $1$,
\item for any addition gate \lal\ with arguments \lbe\ and \lga, the sum of multiplicities of the edges with origin \lal\ is equal to the sum of the multiplicity  of $(\lbe,\lal)$ and the multiplicity of $(\lga,\lal)$,
\item for any multiplication gate \lal\ with arguments \lbe\ and \lga, the sum of multiplicities of the edges with origin \lal\ is equal both to the  multiplicity  of $(\lbe,\lal)$ and to  the multiplicity of $(\lga,\lal)$,
\item for any gate \lal, if the multiplicity of an edge with destination \lal\ is strictly positive, then there must be an edge with origin \lal\ whose multiplicity is strictly positive.
\end{enumerate}
Note that if the circuit is multiplicatively disjoint, since all parse trees are subgraphs of the circuit, the multiplicity of an edge in a parse tree is always at most $1$. In this case a parse tree type contains exactly one parse tree and we could identify parse tree types and parse trees, the characterization above becomes identical to the one given for parse trees.
 
Define the value of a parse tree type as the product, for all input gate \lal, of the label of \lal\ raised to a power equal to the sum of the multiplicities of edges with origin \lal.
The sum of the values of the  parse trees can now be partitioned by parse tree types to express the value computed by $C$:
$$C(\bar{x})=
\sum_{\ltau \text{ a type for } C} 
\abs{\set{T \text{ parse tree of type } \ltau}}\cdot \text{value}(\ltau).$$
 
\end{document}